\documentclass[11pt,fleqn]{article}
\usepackage[utf8]{inputenc}
\usepackage{amsmath,amssymb,graphicx,amsthm,dsfont}
\usepackage{color,soul}
\usepackage{xspace}
\usepackage{enumerate}
\usepackage{mathtools}
\usepackage{bbm}
\usepackage[numbers]{natbib}

\usepackage{multicol}
\usepackage{booktabs}
\usepackage{paralist}

\usepackage{subcaption}
\usepackage[noend,ruled,linesnumbered]{algorithm2e} 
\usepackage{xifthen}

 \usepackage[textsize=tiny,textwidth=1.5cm,linecolor=green!70!black, backgroundcolor=green!10, bordercolor=black,disable=true]{todonotes}%

\usepackage{mdframed}

\usepackage[colorlinks,citecolor=blue!80!black,linkcolor=red!60!black,pagebackref]{hyperref}
\renewcommand*{\backref}[1]{}
\renewcommand*{\backrefalt}[4]{%
\ifcase #1%
\marginpar{\tiny no cite}
\or
 $\rightarrow$~p.~#2.%
\else
  $\rightarrow$~pp.~#2.%
\fi
}

\usepackage{tikz}
\usetikzlibrary{decorations,arrows,petri,topaths,backgrounds,shapes,positioning,fit,calc,decorations.pathreplacing,patterns,intersections,decorations.pathmorphing,matrix,angles,quotes}

\tikzstyle{thickline} = [line width=1.8pt]
\tikzstyle{gl} = [draw, gray]
\tikzstyle{pl} = [draw, orange]
\tikzstyle{ml} = [draw, blue]
\tikzstyle{sett} = [draw, thick, purple]
\tikzstyle{ele} = [draw, thick, green!70!black]
\tikzstyle{nodea} = [fill=green!20!white]
\tikzstyle{nodeb} = [fill=red!20!white]
\makeatletter
\newcommand{\gettikzxy}[3]{%
  \tikz@scan@one@point\pgfutil@firstofone#1\relax
  \edef#2{\the\pgf@x}%
  \edef#3{\the\pgf@y}%
}
\makeatother

\usepackage{cleveref}

\newtheorem{theorem}{Theorem}

\theoremstyle{definition}

\crefname{table}{Table}{Tables}
\crefname{figure}{Figure}{Figures}
\crefname{theorem}{Theorem}{Theorems}
\crefname{definition}{Definition}{Definitions}
\crefname{corollary}{Corollary}{Corollaries}
\crefname{observation}{Observation}{Observations}
\crefname{lemma}{Lemma}{Lemmas}
\crefname{example}{Example}{Examples}
\crefname{reduction}{Reduction}{Reductions}
\crefname{construction}{Construction}{Constructions}
\crefname{subsection}{Subsection}{Subsections}
\crefname{section}{Section}{Sections}
\crefname{proposition}{Proposition}{Propositions}
\crefname{algorithm}{Algorithm}{Algorithms}
\crefname{claim}{Claim}{Claims}

\tikzstyle{alter} = [draw, circle, minimum size=4ex, inner sep=1pt, text centered, align=center]

\tikzstyle{nn} = [draw, circle, inner sep=.7pt,fill=black]
\tikzstyle{dnode} = [nn, rectangle, blue, fill=blue]

\newcommand{\myemph}[1]{{\color{green!40!black}\emph{#1}}}

\newcommand{\RR}{\ensuremath{\mathcal{R}}}
\newcommand{\AAA}{\ensuremath{\mathcal{A}}}

\newcommand{\DWO}{\ensuremath{2\text{-}\mathsf{WO}}}
\newcommand{\TWO}{\ensuremath{3\text{-}\mathsf{WO}}}
\newcommand{\WO}{\ensuremath{\mathsf{WO}}}
\newcommand{\mWO}{\ensuremath{m\text{-}\mathsf{WO}}}

\newcommand{\DEQ}{\ensuremath{2\text{-}\mathsf{EQ}}}

\newcommand{\MP}{\ensuremath{\mathsf{MP}}}

\title{\MP(\RR,\DWO) is Polynomial-Time Solvable}

\author{Jiehua Chen}
\date{31 October 2025\\
  TU Wien}
\begin{document}

\maketitle

\begin{abstract}
  We consider the median procedure~(Barthelemy and Monjardet, 1981) that aggregates a sequence $n$ of binary relations
  from some input class into a single binary relation from some (possibly different) output class,
  minimizing the number of disagreed order pairs.
  We show that if the output class should be a dichotomous weak order~($\DWO$),
  then the problem is polynomial-time solvable.
\end{abstract}

\section{Introduction}

The median procedure (\MP)~\cite{BarMon1981} aggregates $n$ binary relations~$(R_i)_{i\in [n]}$ to another binary relation~$O$ on the same set~$\AAA$ of $m$ alternatives by minimizing the number of disagreements.
Herein, the number of disagreements between two binary relations~$R$ and $R'$ on~$\AAA$ of alternatives is the number of ordered pairs that they disagree, i.e.,
\begin{align*}
  |\MP(R, R')| = |\{(a,b) \in \AAA\times \AAA : ((a,b)\in R \wedge (a,b)\notin R) \vee ((a,b)\notin R \wedge (a,b)\in R)\}|.
\end{align*}
For instance, for the following dichotomous weak orders~$R=\{(a,a),(b,b), (c,c), (a,b), (b,a), (a,c), (b,c)\}$ and $R'=\{(a,a),(b,b), (c,c), (b,c), (c,b), (b,a), (c,a)\}$,
the disagreements are $(a,b)$, $(a,c)$, $(c,a)$, and $(c,b)$. 

We are interested in the computational complexity of determining the minimum number of disagreements for input of certain binary relations and output of another certain type of binary relation.

For input and output that satisfy specific classes there are some known results.
To this end, for some fixed~$t$, let \myemph{$t$-chotomous weak preference order ($t$-\WO)} denote the binary relation which %
corresponds to an ordered partition of the alternatives into $t$ subsets~$(A_1,A_2,\dots,A_t)$ such that for each~$i,i'\in \{1,\dots,t\}$ with $i < i'$ every alternative in~$A_i$ is strictly preferred to every alternative in~$A_{i'}$ and all alternatives in the same subset are tied with each other. 
\citet{Zwicker2018JKKemeny} remarked the following: 
\begin{itemize}
  \item Computing $\MP(\mWO,\mWO)$ is the same as computing the Kemeny score~\cite{Kemeny59}; so it is NP-hard~\cite{BarTovTri1989}. 
  \item Computing $\MP(\DWO,\mWO)$ is the same as ranking by approval score; so it is polynomial-time solvable.
\end{itemize}
Similarly, one can adapt the reduction by \citet{BarTovTri1989} to show that
$\MP(\mWO,\TWO)$ remains NP-hard.
Recently, \citet{CHS2024} consider \myemph{dichotomous equivalence order}~(\myemph{\DEQ}), that is, unordered $2$-partitions~$\{A_1,A_2\}$ with $A_1\cup A_2=A$.
They show that $\MP(\DEQ,\DEQ)$ is NP-hard.
  
In this short note, we show that in contrast to $\DEQ$, computing $\MP(\RR,\DWO)$ is polynomial-time solvable.
Let $\RR=(R_1,R_2,\cdots, R_n)$ be a collection of~$n$ binary relations over~$\AAA$ that are reflexive and let $|\AAA|=m$.
For the sake of reasoning, we assume that each binary relation~$R_i$, $i\in \{1,2,\dots,n\}$, belongs to a distinct voter. 
For each pair of alternatives $a$ and $b$, let \myemph{$N(a,b)$} denote the number of voters that strictly prefer $a$ to $b$ (the corresponding binary relation contains $(a,b)$ but not $(b,a)$),
and \myemph{$E({a,b})$} denote the number of voters that are indifferent between $a$ and $b$
(i.e., the corresponding binary relation contains both $(a,b)$ and $(b,a)$).
Note that by definition, $E(a,b)=E(b,a)$ holds. %

The goal is to find a dichotomous weak order (\DWO) that minimizes the sum of disagreements to~$\RR$. 
Let~$\MP(\RR, \DWO)$ denote the minimum number of disagreements between the input and  any \DWO, that is, %
\begin{align*}
  \MP(\RR,\DWO) \coloneqq \min_{O\in \DWO}\left(\sum_{(a,b)\in O \wedge (b,a)\notin O} 2N(b,a)+E(b,a) + \sum_{(a,b), (b,a)\in O} N(a,b)+N(b,a)\right).
\end{align*}

In the following, we show that determining~$\MP(\RR,\DWO)$ is polynomial-time solvable by reducing to finding a minimum $s$-$t$-cut, which is well-known to be polynomial-time solvable~\cite{CorLeiRivSte2009}.

\section{Construction}
Given~$\RR=(R_1,R_2,\dots,R_n)$, construct a flow network $(G, s, t, w)$ with source $s$ and target $t$ and $m+(m-1)m$ nodes, where each node represents either an alternative: $a$, or an ordered pair of alternatives: $v_{(a,b)}$.

We add the following arcs:
   For each ordered pair of alternatives~$(a,b)$, add 
\begin{itemize}
  \item two arcs $(s,v_{(a,b)})$ and $(v_{(a,b)},t)$ with weight $N(a,b)$ and $N(b,a)$, respectively;
  \item two arcs~$(a,v_{(a,b)})$ and $(v_{(a,b)}, a)$ with a very large weight~$L$ (e.g., $L>(mn)^5$);
  \item two arcs~$(a,b)$ and $(b,a)$ with weights~$E(b,a)$ and $E(a,b)$, respectively; note that $E(a,b)=E(b,a)$.
\end{itemize}
The following is an illustration for the constructed network.

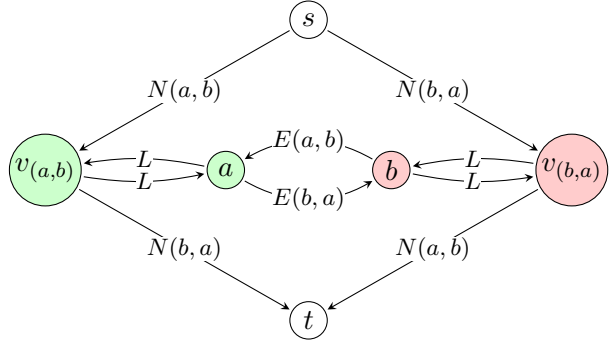
{\centering
  \begin{tikzpicture}[>=stealth, shorten <= 1pt, shorten >= 1pt]
  \node[draw, minimum height = 3ex, inner sep=1pt, circle] at (0,0) (s) {$s$};
  \node[draw, minimum height = 3ex, inner sep=1pt, circle] at (0,-4) (t) {$t$};
  \node[draw, minimum height = 3ex, inner sep=1pt, circle,nodea] at (-3.5,-2) (vab) {$v_{(a,b)}$};
  \node[draw, minimum height = 3ex, inner sep=1pt, circle,nodeb] at (3.5,-2) (vba) {$v_{(b,a)}$};

  \node[draw, minimum height = 3ex, inner sep=1pt, circle,nodea] at (-1.1,-2) (a) {$a$};
  \node[draw, minimum height = 3ex, inner sep=1pt, circle,nodeb] at (1.1,-2) (b) {$b$};

  \foreach \i / \j / \n / \d in {s/vab/{N(a,b)}/0, s/vba/{N(b,a)}/0, vab/t/{N(b,a)}/0, vba/t/{N(a,b)}/0, a/b/{E(b,a)}/30, b/a/{E(a,b)}/30,
  a/vab/L/10, vab/a/L/10, b/vba/L/10, vba/b/L/10} {
    \draw (\i) edge[->, midway, bend right=\d] node [inner sep=0pt,fill=white] {\footnotesize $\n$} (\j);
  }
\end{tikzpicture}
\par}

The goal is to find an $s$-$t$-partition~$(A,B)$ of the nodes with $s\in A$ and $t\in B$ such that the sum of the weights of the arcs from $A$ to $B$ 
\begin{align*}\sum_{(x,y) \text{ is an arc with } x\in A \text{ and } y\in B}w(x,y)\end{align*} is minimized.
We call such partition a \myemph{minimum cut}. 
Note that any \DWO\ corresponds to an ordered partition of the alternatives into two disjoint subsets, and we shall use a minimum cut to determine such an ordered partition.

\begin{theorem}
  If $(A,B)$ is a minimum cut, then the alternatives in~$A$ and $B$ define a DWO with minimum disagreements.
  If $(A',B')$ is a DWO with minimum disagreements,
  then $(A,B)$ is a minimum cut, where $A$ consists of node~$s$, the nodes corresponding to the alternatives in~$A'$, and the nodes~$v_{(x,y)}$ with $x\in A'$ and $y\in B'$, and $B$ consists the remaining nodes.
\end{theorem}

\begin{proof}[Proof sketch.]
 Since the weight between~$v_{(a,b)}$ and $a$ is very large, we can assume without loss of generality that in a minimum cut, $v_{(a,b)}$ and $a$ are always in the same partition.
 Now, consider a partition into $A$ and $B$ with $s\in A$ and $t\in B$.
 There are four cases regarding $a$ and $b$:
\begin{itemize}
  \item If $a\in A$ and $b\in B$ (lower left figure), then the voters who prefer $b$ to $a$ will contribute twice to the weight and the voters who are indifferent between~$a$ and $b$ will contribute once.
  Similarly, if $b\in A$ and $a\in B$, then the voters who prefer $a$ and $b$ will contribute twice to the weight and the voters who are indifferent will contribute once.
  Equivalently, we have the following blue arcs in the cut:
  
  
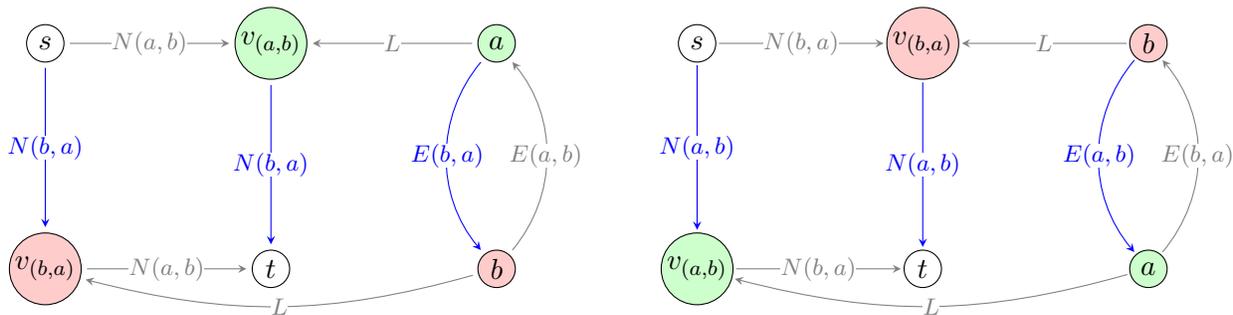
\begin{tikzpicture}[>=stealth, shorten <= 2pt, shorten >= 2pt]
  \node[draw, minimum height = 3ex, inner sep=1pt, circle] at (-3,0) (s) {$s$};
  \node[draw, minimum height = 3ex, inner sep=1pt, circle] at (0,-3) (t) {$t$};
  \node[draw, minimum height = 3ex, inner sep=1pt, circle,nodea] at (3,0) (a) {$a$};
  \node[draw, minimum height = 3ex, inner sep=1pt, circle,nodea] at (0,0) (vab) {$v_{(a,b)}$};
  \node[draw, minimum height = 3ex, inner sep=1pt, circle,nodeb] at (-3,-3) (vba) {$v_{(b,a)}$};
  \node[draw, minimum height = 3ex, inner sep=1pt, circle,nodeb] at (3,-3) (b) {$b$};

  \foreach \i / \j / \n / \d / \c in {s/vab/{N(a,b)}/0/gray, s/vba/{N(b,a)}/0/blue, vab/t/{N(b,a)}/0/blue, vba/t/{N(a,b)}/0/gray, a/b/{E(b,a)}/40/blue, b/a/{E(a,b)}/40/gray, b/vba/L/-15/gray, a/vab/L/0/gray} {
    \draw[\c] (\i) edge[->, midway, bend right=\d] node [inner sep=0pt,fill=white] {\footnotesize $\n$} (\j);
  }
\end{tikzpicture}
\qquad~~\begin{tikzpicture}[>=stealth, shorten <= 1pt, shorten >= 1pt]
  \node[draw, minimum height = 3ex, inner sep=1pt, circle] at (-3,0) (s) {$s$};
  \node[draw, minimum height = 3ex, inner sep=1pt, circle] at (0,-3) (t) {$t$};
  \node[draw, minimum height = 3ex, inner sep=1pt, circle,nodea] at (3,-3) (a) {$a$};
  \node[draw, minimum height = 3ex, inner sep=1pt, circle,nodea] at (-3,-3) (vab) {$v_{(a,b)}$};
  \node[draw, minimum height = 3ex, inner sep=1pt, circle,nodeb] at (0,0) (vba) {$v_{(b,a)}$};
  \node[draw, minimum height = 3ex, inner sep=1pt, circle,nodeb] at (3,0) (b) {$b$};

  \foreach \i / \j / \n / \d / \c in {s/vab/{N(a,b)}/0/blue, s/vba/{N(b,a)}/0/gray, vab/t/{N(b,a)}/0/gray, vba/t/{N(a,b)}/0/blue, a/b/{E(b,a)}/40/gray, b/a/{E(a,b)}/40/blue, b/vba/L/0/gray, a/vab/L/-15/gray} {
    \draw[\c] (\i) edge[->, midway, bend right=\d] node [inner sep=0pt,fill=white] {\footnotesize $\n$} (\j);
  }

\end{tikzpicture}
  \item If $a$ and $b$ are in the same subset (either both in~$A$ or both in~$B$), then the voters who are not indifferent will count once.
  Equivalently, we have the following blue arcs in the cut:
  
  
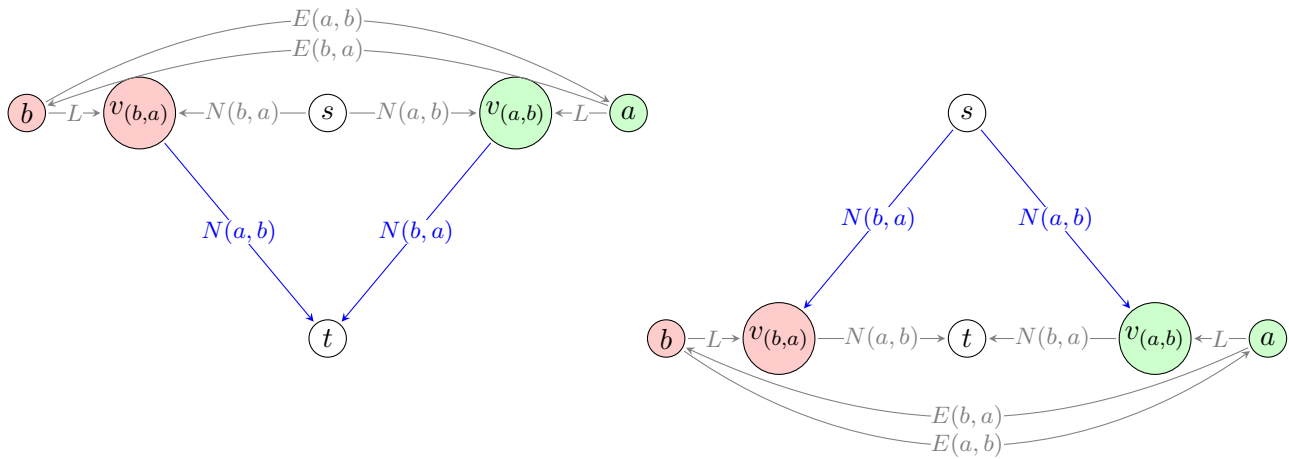
\begin{tikzpicture}[>=stealth, shorten <= 1pt, shorten >= 1pt]
    \node[draw, minimum height = 3ex, inner sep=1pt, circle] at (0,0) (s) {$s$};
    \node[draw, minimum height = 3ex, inner sep=1pt, circle] at (0,-3) (t) {$t$};
    \node[draw, minimum height = 3ex, inner sep=1pt, circle,nodea] at (4,0) (a) {$a$};
    \node[draw, minimum height = 3ex, inner sep=1pt, circle,nodea] at (2.5,0) (vab) {$v_{(a,b)}$};
    \node[draw, minimum height = 3ex, inner sep=1pt, circle,nodeb] at (-2.5,0) (vba) {$v_{(b,a)}$};
    \node[draw, minimum height = 3ex, inner sep=1pt, circle,nodeb] at (-4,0) (b) {$b$};
    
    \foreach \i / \j / \n / \d / \c in {s/vab/{N(a,b)}/0/gray, s/vba/{N(b,a)}/0/gray, vab/t/{N(b,a)}/0/blue, vba/t/{N(a,b)}/0/blue, a/b/{E(b,a)}/20/gray, b/a/{E(a,b)}/-30/gray, b/vba/L/0/gray, a/vab/L/0/gray} {
      \draw[\c] (\i) edge[->, midway, bend right=\d] node [inner sep=0pt,fill=white] {\footnotesize $\n$} (\j);
    }
    
    \begin{scope}[xshift=0.5\textwidth]
      \node[draw, minimum height = 3ex, inner sep=1pt, circle] at (0,0) (s) {$s$};
      \node[draw, minimum height = 3ex, inner sep=1pt, circle] at (0,-3) (t) {$t$};
      \node[draw, minimum height = 3ex, inner sep=1pt, circle,nodea] at (4,-3) (a) {$a$};
      \node[draw, minimum height = 3ex, inner sep=1pt, circle,nodea] at (2.5,-3) (vab) {$v_{(a,b)}$};
      \node[draw, minimum height = 3ex, inner sep=1pt, circle,nodeb] at (-2.5,-3) (vba) {$v_{(b,a)}$};
      \node[draw, minimum height = 3ex, inner sep=1pt, circle,nodeb] at (-4,-3) (b) {$b$};
    
      \foreach \i / \j / \n / \d / \c in {s/vab/{N(a,b)}/0/blue, s/vba/{N(b,a)}/0/blue, vab/t/{N(b,a)}/0/gray, vba/t/{N(a,b)}/0/gray, a/b/{E(b,a)}/-25/gray, b/a/{E(a,b)}/35/gray, b/vba/L/0/gray, a/vab/L/0/gray} {
        \draw[\c] (\i) edge[->, midway, bend right=\d] node [inner sep=0pt,fill=white] {\footnotesize $\n$} (\j);
      }
  \end{scope}
\end{tikzpicture}
\end{itemize}
Since every $s$-$t$-cut~$(A,B)$ corresponds to a possible DWO and the capacity of every minimum cut is equal to the number of disagreements of the corresponding DWO, 
\end{proof}

\bibliographystyle{abbrvnat}
\bibliography{bib}
\end{document}